\documentclass[12pt]{article} 

\usepackage{amsmath,amssymb,amsfonts}
\usepackage{psfrag}
\usepackage{color}
\definecolor{darkblue}{rgb}{0.1,0.1,.7}
\usepackage[colorlinks, linkcolor=darkblue, citecolor=darkblue, urlcolor=darkblue, linktocpage]{hyperref} 
\usepackage[square, comma, compress,numbers]{natbib}
\usepackage[]{amsmath}
\usepackage[]{graphicx}
\usepackage[]{latexsym}
\usepackage{geometry}
\usepackage{amscd}
\usepackage[all,cmtip]{xy}
\usepackage{mathrsfs}
\usepackage[margin=10pt,font=small,labelfont=bf]{caption}
\usepackage{dsdshorthand}
\usepackage{simplewick}
\usepackage{changepage}
\numberwithin{equation}{section}

\renewcommand{\be}{\begin{eqnarray}}
\renewcommand{\ee}{\end{eqnarray}}
\newcommand{\bea}{\begin{eqnarray}}
\newcommand{\eea}{\end{eqnarray}}
 
\newcommand\nn\nonumber 
\newcommand\cS{\mathcal{S}}

\newcommand\SDPAGMP{\texttt{SDPA-GMP}}
\newcommand\SDPB{\texttt{SDPB}}

\newcommand\nmax{n_\mathrm{max}}
\newcommand\numax{\kappa}

\begin{document}

\vspace*{-.6in} \thispagestyle{empty}
\begin{flushright}
\end{flushright}
\vspace{.2in} {\Large
\begin{center}
{\bf A Semidefinite Program Solver for the Conformal Bootstrap \\\vspace{.1in}}
\end{center}
}
\vspace{.2in}
\begin{center}
{\bf 
David Simmons-Duffin} 
\\
\vspace{.2in} 
{\it School of Natural Sciences, Institute for Advanced Study, Princeton, New Jersey 08540}
\end{center}

\vspace{.2in}

\begin{abstract}
We introduce SDPB: an open-source, parallelized, arbitrary-precision semidefinite program solver, designed for the conformal bootstrap.  SDPB significantly outperforms less specialized solvers and should enable many new  computations.  As an example application, we compute a new rigorous high-precision bound on operator dimensions in the 3d Ising CFT, $\De_\s=0.518151(6)$, $\De_\e=1.41264(6)$.
\end{abstract}

\newpage

\tableofcontents

\newpage

\section{Introduction}
\label{sec:intro}

In \cite{Rattazzi:2008pe}, Rattazzi, Rychkov, Tonni, and Vichi showed how to bound Conformal Field Theory (CFT) observables using convex optimization. Shockingly, the resulting bounds are sometimes saturated by actual CFTs, allowing high-precision computations of nonperturbative quantities \cite{Rychkov:2009ij,Rychkov:2011et,ElShowk:2012ht,El-Showk:2014dwa,ElShowk:2012hu,Kos:2013tga,Beem:2013qxa,Kos:2014bka}.  This modern incarnation of the conformal bootstrap \cite{Ferrara:1973yt, Polyakov:1974gs} has been applied to numerous theories, and the list is growing \cite{Caracciolo:2009bx,Poland:2010wg,Rattazzi:2010gj,Rattazzi:2010yc,Vichi:2011ux,Poland:2011ey,Liendo:2012hy,El-Showk:2013nia,Alday:2013opa,Gaiotto:2013nva,Bashkirov:2013vya,Berkooz:2014yda,Nakayama:2014lva,Nakayama:2014yia,Chester:2014fya,Caracciolo:2014cxa,Alday:2014qfa,Nakayama:2014sba,Chester:2014mea,Bae:2014hia,Beem:2014zpa,Chester:2014gqa}.

With recent analytical work \cite{Costa:2011mg,Costa:2011dw,SimmonsDuffin:2012uy,Dymarsky:2013wla,Beem:2013sza,Fitzpatrick:2014oza,Khandker:2014mpa}, we now understand in principle how to formulate many interesting bootstrap calculations, like studies of four-point functions of scalars, fermions, conserved currents, stress-tensors, and mixed correlators combining all of these ingredients.  Such studies may shed light on the classification of critical points in condensed matter systems, the conformal windows of 4d gauge theories, the landscape of AdS string vacua, and more.  Each study inevitably culminates in an optimization problem that only can be solved numerically at present.

\SDPB\ is a custom optimizer with three purposes:
\begin{enumerate}

\item To enable the next generation of bootstrap studies involving conserved currents, stress tensors, and other complex ingredients.

\item To help improve predictions for CFTs already isolated with the bootstrap, like the 3d Ising CFT, 3d $O(N)$ vector models, and others.

\item To demonstrate the potential for dramatic improvement in current numerical bootstrap techniques and to encourage researchers in numerical optimization and computer science to contribute ideas and expertise.

\end{enumerate}

Custom optimizers have improved bootstrap calculations in the past.  For example in \cite{El-Showk:2014dwa}, a custom-written linear program solver enabled a new high-precision calculation of critical exponents in the 3d Ising CFT, surpassing what was possible with out-of-the-box solvers like  \texttt{Mathematica} \cite{Mathematica} (used in the original work \cite{Rattazzi:2008pe}), \texttt{GLPK} \cite{glpk}, and \texttt{CPLEX} \cite{cplex}.\footnote{The solver in \cite{El-Showk:2014dwa} is more accurately called a semi-infinite program solver.}

Unfortunately, linear programming is not applicable to systems of multiple correlators or operators with spin.  These more complicated cases can be attacked with semidefinite programming \cite{Poland:2011ey,Kos:2014bka}, for which previous studies have relied on the solvers \texttt{SDPA} \cite{SDPA,SDPA2} and \SDPAGMP\ \cite{SDPAGMP}.  The study \cite{Kos:2014bka}, in particular, pushed \SDPAGMP\ to its limits, with each optimization taking up to 2 weeks.  By contrast, \SDPB\ can perform the same optimization in 1-3 CPU-hours, or 4-12 minutes on a 16-core machine.

A general bootstrap problem can be approximated as a particular type of semidefinite program called a ``polynomial matrix program" (PMP).  \SDPB\ implements a variant of the well-known primal-dual interior point method for semidefinite programming \cite{doi:10.1137/1.9781611970791,alizadeh,citeulike:3473686,Vandenberghe:1995:PPR:213565.213573}, specialized for PMPs.  Specialization and parallelization are \SDPB's advantages.  We are optimistic that better designs and algorithms can be brought to bear in the future.

In section~\ref{sec:designofSDPB}, we describe PMPs and the design of \SDPB.  This discussion is mostly independent of the conformal bootstrap and should be comprehensible to readers without a physics background.  

In section~\ref{sec:applicationtobootstrap}, as an application of \SDPB, we set a new world-record for precision of critical exponents in the 3d Ising CFT, using multiple correlators as in \cite{Kos:2014bka}.  Readers interested solely in physics can skip to this section.  We conclude in section~\ref{sec:discussion}.

\SDPB\ is open source and available online at \href{https://github.com/davidsd/sdpb}{\texttt{https://github.com/davidsd/sdpb}}.

\section{Design of \SDPB}
\label{sec:designofSDPB}

\subsection{Polynomial Matrix Programs}
\label{sec:PMP}

\SDPB\ solves the following type of problem, which we call a {\it polynomial matrix program} (PMP).  Consider a collection of symmetric polynomial matrices
\be
M_j^n(x) &=& \begin{pmatrix}
P_{j,11}^{n}(x) & \dots & P_{j,1m_j}^{n}(x)\\
\vdots & \ddots & \vdots\\
P_{j,m_j1}^{n}(x) & \dots & P^{n}_{j,m_jm_j}(x)
\end{pmatrix}
\ee
labeled by $0 \leq n \leq N$ and $1 \leq j \leq J$,
where each element $P_{j,rs}^{n}(x)$ is a polynomial.  
Given $b\in \R^N$, we would like to
\be
\label{eq:PMPconstraint}
\begin{array}{ll}
\textrm{maximize} & b\.y\quad\textrm{over}\quad y \in \R^N,\\
\textrm{such that} & M^0_j(x)+\sum_{n=1}^{N} y_n M^n_j(x) \succeq 0 \quad \textrm{for all $x\geq 0$ and $1 \leq j \leq J$}.
\end{array}
\ee
The notation $M\succeq 0$ means ``$M$ is positive semidefinite."

As we review in section~\ref{sec:applicationtobootstrap}, a wide class of optimization problems from the conformal bootstrap can be written in this form.  Conveniently, PMPs can be translated into semidefinite programs (SDPs) and solved using interior point methods.  In the next few subsections, we perform this translation and describe an interior point algorithm for solving general SDPs.  Subsequently, we make this algorithm more efficient by exploiting special structure in PMPs.

\subsection{Translating PMPs into SDPs}
\label{sec:translationPMPtoSDP}

Let us begin by translating the PMP (\ref{eq:PMPconstraint}) into a more standard semidefinite program of the following form:
\be
\label{eq:traditionalSDP}
\begin{array}{ll}
 \textrm{maximize} & \Tr(CY) + b \. y \quad \textrm{over} \quad y\in \R^N,\ Y\in \cS^K, \\
\textrm{such that} & \Tr(A_* Y)+By = c,\ \textrm{and}\\
& Y \succeq 0,
\end{array}
\ee
where 
\be
c &\in& \R^P, \nn\\
B &\in& \R^{P\x N}, \nn\\
A_1,\dots,A_P,C &\in& \cS^K.
\ee
Here, $\cS^K$ is the space of $K\x K$ symmetric real matrices, and $\Tr(A_* Y)$ denotes the vector $(\Tr(A_1 Y),\dots,\Tr(A_P Y))\in\R^P$.  The SDP (\ref{eq:traditionalSDP}) is similar to those treated by solvers like \SDPAGMP, except that it includes the variables $y\in \R^N$, called ``free variables" because they are not constrained to be positive. The matrix $C$ will eventually be set to zero, but we include it for generality.

The first step is to relate positive semidefiniteness of polynomial matrices to positive semidefiniteness of a single matrix $Y$.  Let $q_0(x),q_1(x),\dots,$ be a collection of polynomials with degrees $0,1,\dots$ (for example, monomials $q_m(x)=x^m$), and define the vector $\vec q_\delta(x)=(q_0(x),\dots,q_\delta(x))$.  
We call $q_m(x)$ a ``bilinear basis" because products $q_m(x)q_n(x)$ span the space of polynomials.  In particular, any polynomial $P(x)$ of degree $d$ can be written
\be
P(x) &=& \Tr_{\R^{\delta_1+1}}(Y_1 Q_{\delta_1}(x)) + x \Tr_{\R^{\delta_2+1}}(Y_2 Q_{\delta_2}(x)),
\ee
where
\be
Q_\delta(x) &\equiv& \vec q_\delta(x)\vec q_\delta(x)^T,\nn\\
\delta_1 &\equiv& \lfloor d/2 \rfloor,\nn\\
\delta_2 &\equiv & \lfloor (d-1)/2\rfloor,
\ee
and $Y_1,Y_2$ are (underdetermined) symmetric matrices
\be
Y_1 &\in& \cS^{\de_1+1},\nn\\
Y_2 &\in& \cS^{\de_2+1}.
\ee
For a symmetric $m\x m$ polynomial matrix $M(x)$ of degree $d$, we can apply this construction to each element,
\be
\label{eq:bilinearbasisforM}
M(x) &=& \Tr_{\R^{\delta_1 + 1}}(Y_1 (Q_{\delta_1}(x) \otimes \mathbf{1}_{m\x m})) + x \Tr_{\R^{\delta_2 + 1}}(Y_2 (Q_{\delta_2}(x) \otimes \mathbf{1}_{m\x m})),\nn\\
Y_1 &\in& \cS^{m(\delta_1+1)},\nn\\
Y_2 &\in& \cS^{m(\delta_2+1)}.
\ee
where now $Y_1$ acts on $\R^{\delta_1+1}\otimes \R^m$ and $Y_2$ similarly acts on $\R^{\delta_2+1}\otimes \R^m$, and each trace is over the first tensor factor.

The translation of PMPs into SDPs relies on the following theorem:
\begin{theorem}
\label{thm:positivesemidefiniteness}
$M(x)$ is positive semidefinite for $x\in \R^+$ if and only if it can be written in the form (\ref{eq:bilinearbasisforM}) for some positive semidefinite $Y_1$ and $Y_2$.
\end{theorem}
\begin{proof}
One direction is simple: choose a vector $v\in \R^m$ and consider the pairing $v^T M(x) v=\Tr(Y_1(Q_{\delta_1}(x)\otimes vv^T))+\Tr(Y_2(xQ_{\delta_2}(x)\otimes vv^T))$.  Suppose $Y_1,Y_2$ are positive semidefinite and $x\geq 0$.  Then since $Q_{\delta_1}(x)\otimes vv^T$ and $xQ_{\delta_2}(x)\otimes vv^T$ are both positive semidefinite, it follows that $v^T M(x)v\geq 0$.
The other direction is less trivial.  It has been proven directly in \cite{HanselkaUnpublished} and also follows as a consequence of the ``Biform Theorem" of \cite{biformtheorem}, using the substitution $x=y^2$ and results of \cite{4435026}.\footnote{We thank Pablo Parrilo for pointing this out.}
\end{proof}

Theorem~\ref{thm:positivesemidefiniteness} lets us rewrite our PMP constraints (\ref{eq:PMPconstraint}) in terms of a collection of positive semidefinite matrices $Y_1,\dots,Y_{2J}$. We equate each polynomial matrix $M_j^n(x)$ to an expression of the form (\ref{eq:bilinearbasisforM}).  Individual matrix elements $P^{n}_{j,rs}$ can be isolated by taking the trace over $\R^{m_j}$ with symmetrized unit matrices
\be
(E^{rs})_{ij} &\equiv& \frac 1 2 (\de^r_i \de^s_j+\de^s_i \de^r_j).
\ee
This gives a set of polynomial equalities
\be
\label{eq:polynomialequality}
P^{0}_{j,rs}(x) + \sum_n y_n P^{n}_{j,rs}(x) &=& \Tr(Y_{2j-1} (Q_{\de_{j1}}(x)\otimes E^{rs})) + \Tr(Y_{2j} ( xQ_{\de_{j2}}(x)\otimes E^{rs})),\nn\\
d_j &\equiv& \max_{n=0}^N\, \mathrm{deg}(M_j^n(x)),\nn\\
\de_{j1}&\equiv&\lfloor d_j/2\rfloor,\nn\\
\de_{j2}&\equiv&\lfloor (d_j-1)/2\rfloor.
\ee
Equality between polynomials of degree $d_j$ is equivalent to equality at $d_j+1$ points.  Thus, evaluating (\ref{eq:polynomialequality}) at points $x_0,\dots,x_{d_j}$, we obtain a set of affine relations between the $y_n$ and $Y_j$,\footnote{We could alternatively match coefficients on both sides of (\ref{eq:polynomialequality}) as in \cite{Poland:2011ey,Kos:2013tga,Kos:2014bka}.  We will see in subsection~\ref{sec:computingschur} why pointwise evaluation is preferable.}
\be
\label{eq:linearPMPequalities}
P^{0}_{j,rs}(x_k)+\sum_n y_n P^{n}_{j,rs}(x_k) &=& \Tr(Y_{2j-1} ( Q_{\de_{j1}}(x_k)\otimes E^{rs})) + \Tr(Y_{2j} (x_k Q_{\de_{j2}}(x_k) \otimes E^{rs}))\nn\\
&& 0 \leq r\leq s < m_j,\quad k=0,\dots,d_j.
\ee

Let us group the $Y_j$'s into a single block-diagonal positive semidefinite matrix
\be
\label{eq:blockstructureforY}
Y = \begin{pmatrix}
Y_1 & 0 & \cdots & 0\\
 0 & Y_2 & \cdots & 0\\
 \vdots & \vdots & \ddots & \vdots\\
 0 & 0 & \cdots & Y_{2J}
\end{pmatrix}.
\ee
The constraints (\ref{eq:linearPMPequalities}) now take the form
\be
\Tr(A_p Y) + (B y)_p = c_p,
\ee
where $p$ runs over all tuples $(j,r,s,k)$ satisfying $0\leq r \leq s < m_j$, and $0\leq k \leq d_j$.  The matrices $A_p, B,C$ and vector $c$ are given by
\be
\label{eq:ApforPMP}
A_{(j,r,s,k)} &=&
\begin{pmatrix}
0 & \cdots & 0 & 0 & \cdots & 0\\
\vdots & \ddots & \vdots & \vdots & &  \vdots\\
0 & \cdots & Q_{\de_{j1}}(x_k)\otimes E^{rs} & 0 & \cdots & 0\\
0 & \cdots & 0 & x_kQ_{\de_{j2}}(x_k)\otimes E^{rs} & \cdots & 0\\
\vdots & & \vdots & \vdots & \ddots & \vdots\\
0 & \cdots & 0 & 0 & \cdots & 0
\end{pmatrix},\\
\label{eq:BforPMP}
B_{(j,r,s,k),n} &=& -P^{n}_{j,rs}(x_k),\\
\label{eq:CforPMP}
c_{(j,r,s,k)} &=& P^{0}_{j,rs}(x_k),\\
C &=& 0.
\ee

This completes the translation of our PMP into an SDP (\ref{eq:traditionalSDP}).  Note that the matrices $A_p$ are far from generic.  Exploiting this fact will help us solve PMPs much more efficiently than a generic SDP.  \SDPB\ is specifically designed to solve SDPs with constraint matrices of the form
\be
\label{eq:specialconstraintmatrices}
A_{(j,r,s,k)} &=& \sum_{b\,\in\, \texttt{blocks}_j} \cB_b(\vec v_{b,k} \vec v_{b,k}^T \otimes E^{rs}),
\ee
where $\cB_b(M)$ denotes the block-diagonal matrix with $M$ in the $b$-th block and zeros everywhere else,
\be
\cB_b(M) &\equiv&
\begin{pmatrix}
0 & \cdots & 0 & \cdots & 0  \\
\vdots & \ddots & \vdots & & \vdots  \\
0 & \cdots & M & \cdots & 0  \\
\vdots & & \vdots & \ddots & \vdots \\
0 & \cdots & 0 & \cdots & 0 & \\
\end{pmatrix}
\hspace{-0.4in}
\begin{array}{c}
\left.\phantom{\begin{pmatrix}
1\\
2\\
\\
\end{pmatrix}}\right\}b\\
\phantom{ \footnotesize {\begin{pmatrix}
3\\
4\\
\\
\end{pmatrix}}}
\end{array},
\ee
and the sets $\texttt{blocks}_j$ are disjoint for different $j$.  In the example above, we have
\be
\texttt{blocks}_j &=& \{2j-1, 2j\},\nn\\
\vec v_{2j-1,k} &=& \vec q_{\delta_{j1}}(x_k),\nn\\
\vec v_{2j,k} &=& x_k^{1/2} \vec q_{\delta_{j2}}(x_k).
\ee

\subsection{Semidefinite Program Duality}

Duality plays an important role in our interior point algorithm, so let us briefly review it.  
The problem (\ref{eq:traditionalSDP}) is related by duality to the following ``primal" optimization problem:
\be
\label{eq:primaldualproblems}
\begin{array}{rll}
\mathcal{P}: & \textrm{minimize} & c\.x \quad \textrm{over}\quad x\in \R^P,\ X\in \cS^K,\\
& \textrm{such that} & X= \sum_{p=1}^P A_p x_p - C,\\
& &  B^T x= b,\\
& &  X\succeq 0.
\end{array}
\ee
We refer to the problem (\ref{eq:primaldualproblems}) as $\cP$ for ``primal" and (\ref{eq:traditionalSDP}) as $\cD$ for ``dual."  We say that $\cP$ is ``feasible" if there exist $x,X$ satisfying the constraints~(\ref{eq:primaldualproblems}).  Similarly, $\cD$ is feasible if there exist $y,Y$ satisfying the constraints~(\ref{eq:traditionalSDP}).  The ``duality gap" is defined as the difference in primal and dual objective functions, $c\.x - \Tr(CY)-b\.y$.

For our purposes, the statement of duality is as follows:

\begin{theorem}[Semidefinite Program Duality]
Given a feasible point $(x,X)$ of $\cP$ and a feasible point $(y,Y)$ of $\cD$, the duality gap is nonnegative.  If the duality gap vanishes, then $(x,X)$ and $(y,Y)$ are each optimal solutions of $\cP$ and $\cD$, and furthermore $XY=0$.
\end{theorem}

\begin{proof}
Suppose we have feasible solutions $(x,X)$ and $(y,Y)$.  The duality gap is given by
\be
c\.x - \Tr(C Y) - b\.y &=& c\.x - \Tr\p{\p{\sum_p A_p x_p-X}Y} - b\.y\nn\\
&=& c\.x  + \Tr(XY)- x\.\Tr(A_* Y) -b\.y\nn\\
&=& c\.x+\Tr(XY) - x\.(\Tr(A_* Y) + By)\nn\\
&=& \Tr(XY) \geq 0,
\ee
where nonnegativity follows because $X$ and $Y$ are positive semidefinite.  Now suppose $\Tr(XY)$ vanishes.  Clearly this implies $XY=0$ identically (this condition is called ``complementarity").  Because $c\.x$ is bounded from below by the dual objective $\Tr(CY)+b\.y$ and also equal to the dual objective, the point $(x,X)$ must be optimal.  Similarly, since $\Tr(CY)+b\.y$ is bounded from above by the primal objective $c\.x$ and also equal to the primal objective, the point $(y,Y)$ must be optimal as well.
\end{proof}

Unlike in linear programming, there is no guarantee that either $\cP$ or $\cD$ will attain their respective optima, or that the duality gap will vanish.  For this, we need additional regularity assumptions.  One of them is Slater's condition, which says that the duality gap vanishes if there exist {\em strictly} feasible solutions to the primal and dual constraints --- i.e. solutions where $X,Y\succ 0$ are positive-definite.  
Slater's condition is generic in the sense that a small perturbation of a feasible but not strictly-feasible problem will typically produce a strictly-feasible problem.

\subsection{An Interior Point Method}

The idea behind primal-dual interior point methods is to solve both the primal and dual equations simultaneously to find an optimal point $q=(x,X,y,Y)$.
As we saw in the previous subsection, the optimum is (generically) achieved by a pair of positive semidefinite matrices $X, Y$ satisfying the ``complementarity" condition  $XY=0$.  Most algorithms work by deforming this condition to 
\be
\label{eq:complementarity}
XY=\mu I
\ee
for some nonzero $\mu$, where $I$ is the identity matrix.  The constraints together with (\ref{eq:complementarity}) then have a unique family of solutions called the ``central path:" $q(\mu)=(x(\mu),X(\mu),y(\mu),Y(\mu))$ indexed by $\mu\in \R^+$.  By following the central path from positive $\mu$ towards $\mu=0$, we can find the optimum of the original problem.

In practice, instead of moving precisely along the central path, we use the following strategy. Consider an initial point $q=(x,X,y,Y)$ such that $X,Y$ are positive semidefinite.  Our goal is to decrease $\Tr(XY)$ and move towards the constraint locus while maintaining positive semidefiniteness.
\begin{itemize}
\item Set $\mu = \beta \Tr(XY)/K$ for some $\beta<1$, where $K$ is the number of rows of $X$.
\item Use Newton's method to compute a direction $dq=(dx,dX,dy,dY)$ towards the central path with the given $\mu$.
\item Take a step along $dq$, taking care not to violate the positive semidefiniteness of $X,Y$.  This should result in a reduction of $\Tr(XY)$ by roughly a factor of $\beta$.
\item Repeat.
\end{itemize}
This is essentially Newton's method with a moving target.

An important advantage of this method is that the initial starting point $(x,X,y,Y)$ need not satisfy any of the equality constraints in (\ref{eq:traditionalSDP}) and (\ref{eq:primaldualproblems}).  As long as we start with positive semidefinite $X,Y$, and the problem is feasible, the above method will converge to a point where the equality constraints are satisfied.

\subsubsection{Newton Search Direction}

Let us describe a single Newton step in more detail.  The direction $dq$ is defined by replacing $q\to q+dq$ and solving the constraint equations (\ref{eq:traditionalSDP}, \ref{eq:primaldualproblems}) and complementarity equation (\ref{eq:complementarity}) at linear order in $dq$,
\be
\label{eq:searchstep}
X+dX &=& \sum_i A_i(x_i+dx_i) - C,\nn\\
B^T(x+dx) &=& b,\nn\\
\Tr(A_* (Y+dY))+B(y+dy) &=& c,\nn\\
XY + X dY + dX Y &=& \mu I.
\ee
The residues
\be
\label{eq:residues}
P &\equiv& \sum_i A_i x_i - X - C, \nn\\
p &\equiv& b - B^T x, \nn\\
d &\equiv& c - \Tr(A_* Y) - B y, \nn\\
R &\equiv& \mu I - XY,
\ee
measure the failure of the current point to satisfy the constraints.  These residues will decrease with each Newton step.  The linearized equations (\ref{eq:searchstep}) can then be written
\be
\label{eq:searchstepsolutionblock}
\begin{pmatrix}
S & -B\\
B^T & 0
\end{pmatrix}
\begin{pmatrix}
dx\\
dy
\end{pmatrix} &=&
\begin{pmatrix}
-d-\Tr(A_*Z)\\
p
\end{pmatrix}, \\
\label{eq:dX}
dX &=& P + \sum_i A_i dx_i,\\
\label{eq:dYWrong}
dY &=& X^{-1}(R- dX Y),
\ee
where $Z = X^{-1} (PY - R)$ and the ``Schur complement" matrix $S$ is defined by
\be
S_{ij} &=& \Tr(A_i X^{-1} A_j Y).
\ee
To find the search direction, we first solve (\ref{eq:searchstepsolutionblock}) for $dx$, $dy$, and then plug into (\ref{eq:dX}) and $(\ref{eq:dYWrong})$ to determine $dX,dY$.  Na\"ively applying (\ref{eq:dYWrong}) leads to a $dY$ that is not necessarily symmetric, which would take us outside of the domain of definition of $Y$.  Several solutions to this problem have been proposed \cite{Todd98astudy}.  Our approach, following \cite{citeulike:3473686,Kojima:1997:IMM:588890.589063,Monteiro96primal-dualpath-following}, will be to symmetrize $dY$ by hand, replacing (\ref{eq:dYWrong}) with 
\be
\label{eq:dY}
\widehat{dY} &=& X^{-1}(R- dX Y),\nn\\
dY &=& \frac 1 2 \p{\widehat {dY} + \widehat {dY}^T}.
\ee

\subsubsection{Mehrotra Predictor-Corrector Trick}
\label{sec:mehrotra}

The most expensive operations in the search direction calculation are forming the Schur complement matrix $S$ and solving the Schur complement equation (\ref{eq:searchstepsolutionblock}).  We'd like to perform them as rarely as possible.  A simple modification to the na\"ive Newton's method, due to Mehrotra \cite{doi:10.1137/0802028}, allows us to get closer to the central path while reusing $S$ and most of the work done in solving  (\ref{eq:searchstepsolutionblock}).

The rough idea is to solve the constraint and complementarity equations at higher order. This proceeds in two steps, called the ``predictor" and ``corrector" steps, respectively.  For the predictor step, we compute a direction as described above, which we call $dq_p=(dx_p,dX_p,dy_p,dY_p)$.  We then replace the linearized complementarity equation (\ref{eq:searchstep}) with
\be
\label{eq:correctorcomplementarity}
XY+X dY + dX Y + dX_p dY_p = \mu I
\ee
and re-solve to obtain a corrector direction $dq_c$.  In the corrector step, we may (optionally) use a smaller deformation parameter $\mu \to \mu_c$ to get closer to $\Tr(XY)= 0$.  Note that the replacement of (\ref{eq:searchstep}) with (\ref{eq:correctorcomplementarity}) does not change the Schur complement matrix $S_{ij}$, so it can be reused, together with any matrix decompositions performed in solving (\ref{eq:searchstepsolutionblock}).  Altogether, the corrector step amounts to simply replacing
\be
R \to \mu_c I - X Y - dX_p dY_p
\ee
before solving (\ref{eq:searchstepsolutionblock}), (\ref{eq:dX}), and (\ref{eq:dY}).

\subsubsection{Termination Conditions}

We say a point $q$ is ``primal feasible" if the residues $p,P$ are sufficiently small.  Similarly, the solution is ``dual feasible" if the residue $d$ is sufficiently small.  The precise conditions are
\be
\begin{array}{rrcccl}
\textrm{primal feasible:} &  \texttt{primalError} &\equiv& \max_{i,j}\{|p_i|, |P_{ij}|\} &<& \texttt{primalErrorThreshold};\\
\textrm{dual feasible:} & \texttt{dualError} &\equiv& \max_i\{|d_i|\} &<& \texttt{dualErrorThreshold},
\end{array}\nn\\
\ee
where $\texttt{primalErrorThreshold}\ll 1$ and $\texttt{dualErrorThreshold} \ll 1$ are parameters chosen by the user.

An optimal point should be both primal and dual feasible, and have (nearly) equal primal and dual objective values.  Specifically, let us define $\texttt{dualityGap}$ as the normalized difference between the primal and dual objective functions
\be
\texttt{dualityGap} &\equiv& \frac{|\texttt{primalObjective} - \texttt{dualObjective}|}{\max\{1, |\texttt{primalObjective} + \texttt{dualObjective}|\}}, \nn\\
\texttt{primalObjective} &\equiv& c\. x, \nn\\
\texttt{dualObjective} &\equiv& \Tr(CY)+b\.y.
\ee
A point is considered ``optimal" if
\be
\texttt{dualityGap} &<& \texttt{dualityGapThreshold},
\ee
where $\texttt{dualityGapThreshold} \ll 1$ is chosen by the user.

\subsubsection{Complete Algorithm}
\label{sec:completealgorithm}

Our complete interior point algorithm is as follows.

\begin{enumerate}

\item Choose an initial point $q=(x,X,y,Y)=(0,\Omega_\cP I,0,\Omega_\cD I)$ where $\Omega_\cP$ and $\Omega_\cD$ are real and positive.  This point probably does not satisfy the constraints.

\item \label{startofrecursion}
 Compute the residues (\ref{eq:residues}) and terminate if $q$ is simultaneously primal feasible, dual feasible, and optimal.  (Sometimes we may wish to use a different termination criterion, see below.)

\item Let $\mu=\Tr(XY)/K$ and compute the predictor deformation parameter $\mu_p = \beta_p \mu$ where
\be
\beta_p = \begin{cases}
0 & \textrm{if $q$ is primal and dual feasible};\\
\beta_{\textrm{infeasible}} & \textrm{otherwise}. \\
\end{cases}
\ee
Here, $\beta_{\textrm{infeasible}} \in (0,1)$ is chosen by the user.

\item Compute the predictor search direction $dq_p$ by solving eqs.~(\ref{eq:searchstepsolutionblock}, \ref{eq:dX}, \ref{eq:dY}) with $R=\mu_p I - XY$.

\item Compute the corrector deformation parameter $\mu_c = \beta_c \mu$ as follows.  Let $r=\Tr((X+dX_p)(Y+dY_p))/(\mu K)$ and set
\be
\beta &=& \begin{cases}
r^2 & \textrm{if}\ r<1;\\
r & \textrm{otherwise},
\end{cases}\\
\beta_c &=& \begin{cases}
\min(\max(\beta_\textrm{feasible}, \beta), 1) & \textrm{if $q$ is primal and dual feasible};\\
\max(\beta_\textrm{infeasible}, \beta) & \textrm{otherwise},
\end{cases}
\ee
where $\beta_\textrm{feasible}\in (0,1)$ is a parameter chosen by the user.
This choice of $\beta_c$ is modeled after the one in \texttt{SDPA}.

\item Compute the corrector search direction $dq_c$ by solving  eqs.~(\ref{eq:searchstepsolutionblock}, \ref{eq:dX}, \ref{eq:dY}) with $R=\mu_c I - XY - dX_p dY_p$.

\item Determine the primal and dual step lengths 
\be
\label{eq:steplengths}
\alpha_\cP &=& \min(\gamma\alpha(X,dX_c),1),\nn\\
\alpha_\cD &=& \min(\gamma\alpha(Y,dY_c),1),
\ee
where $\alpha(M,dM)$ is the largest positive real number such that $M+\alpha(M,dM) dM$ is positive semidefinite, and $\gamma \in (0,1)$ is a parameter chosen by the user.

\item Replace
\be
\label{eq:takingnewtonstep}
x &\to& x+\alpha_\cP dx_c,\nn\\
X &\to& X+\alpha_\cP dX_c,\nn\\
y &\to& y+\alpha_\cD dy_c,\nn\\
Y &\to& Y+\alpha_\cD dY_c,
\ee
and go to step \ref{startofrecursion}.  Note that the replacement (\ref{eq:takingnewtonstep}) is guaranteed to preserve positive semidefiniteness of $X$ and $Y$.

\end{enumerate}

If the current point is close enough to a primal (or dual) feasible region, the step-length $\alpha_\cP$ ($\alpha_\cD$) in (\ref{eq:steplengths}) can be exactly 1.  When this occurs, the replacement (\ref{eq:takingnewtonstep}) will exactly solve the primal (dual) equality constraints, up to numerical errors.  This follows from linearity of the equality constraints, together with the fact that symmetrizing $Y$ in (\ref{eq:dY}) has no effect on the constraints.  In cases where we only care about primal or dual feasibility, the iteration can be stopped here, see section~\ref{sec:example3dIsing}.

\subsection{Specialization to Polynomial Matrix Programs}

As mentioned in section~\ref{sec:mehrotra}, the most expensive part of the search direction calculation is computing the Schur complement matrix $S_{pq}=\Tr(A_p X^{-1} A_q Y)$ and inverting (\ref{eq:searchstepsolutionblock}) to obtain $dx,dy$.  In this section, we will specialize to PMPs, and study how these calculations can be made more efficient.  For similar optimizations, see \cite{4434343}.

\subsubsection{Block Structure of the Schur Complement Matrix}

Recall that for PMPs, the matrices $A_p$ are given by (\ref{eq:specialconstraintmatrices}) with the index $p$ running over tuples $(j,r,s,k)$ satisfying $0\leq r \leq s < m_j$, and $0\leq k \leq d_j$.  Since $X$ and $Y$ have the block structure (\ref{eq:blockstructureforY}), the Schur complement matrix $S_{p_1p_2}$ is block-diagonal: it has nonzero entries only if $j_1=j_2$,
\be
S &=& \begin{pmatrix}
S^{(1)} & 0 & \cdots & 0\\
 0 & S^{(2)} & \cdots & 0\\
 \vdots & \vdots & \ddots & \vdots\\
 0 & 0 & \cdots & S^{(J)}
 \end{pmatrix}.
\ee
The size of the $j$-th block $\dim S^{(j)}$ is equal to the number of choices for $(r,s,k)$,
\be
\dim S^{(j)} &=& \frac{m_j(m_j+1)}{2}(d_j+1).
\ee

Now consider equation (\ref{eq:searchstepsolutionblock}),
\be
\label{eq:searchstepsolutionblocktwo}
T\begin{pmatrix}
dx\\
dy
\end{pmatrix}
&=&
\begin{pmatrix}
-d-\Tr(A_* Z)\\
p
\end{pmatrix},\quad \textrm{where} \quad
T \equiv \begin{pmatrix}
S & -B\\
B^T & 0
\end{pmatrix}.
\ee
We could solve it using an LU (lower triangular $\x$ upper triangular) decomposition of $T$, but this is extremely expensive, taking cubic time in $\dim T = \sum_j \dim S^{(j)}+N$.  

We should use the block structure of $T$ to our advantage. Let $S=L L^T$ be a Cholesky decomposition of $S$ (which can be computed efficiently for each block $S^{(j)}=L^{(j)} L^{(j)T}$), and consider the decomposition
\be
\label{eq:naiveblockdecomposition}
\begin{pmatrix}
S & -B\\
B^T & 0
\end{pmatrix}
&=&
\begin{pmatrix}
L & 0\\
B^T L^{-T} & \mathbf{1}
\end{pmatrix}
\begin{pmatrix}
\mathbf{1} & 0\\
0 & B^TL^{-T}L^{-1}B
\end{pmatrix}
\begin{pmatrix}
L^T & -L^{-1}B\\
0 & \mathbf{1}
\end{pmatrix}.
\ee
The outer matrices on the right hand side are triangular, and can be solved efficiently by forward/backward-substitution.  Meanwhile, the matrix $Q\equiv B^T L^{-T} L^{-1} B$ typically has a much smaller dimension than $S$, $\dim Q=N \ll \dim T$, so the middle block-matrix can be easily solved using a Cholesky decomposition.\footnote{The matrix $Q$, is often called the ``Schur complement" in the block matrix decomposition~(\ref{eq:naiveblockdecomposition}).  We will continue to use the words ``Schur complement" to refer to $S$ (which is the Schur complement of a different block matrix system).  Hopefully this will not cause confusion.}

Unfortunately, the decomposition (\ref{eq:naiveblockdecomposition}) is numerically unstable when $S$ is ill-conditioned.  Indeed, suppose $S$ has a very small eigenvalue (so $L$ does too), while the full matrix $T$ does not.  Then quantities like $L^{-1}B$ that appear in (\ref{eq:naiveblockdecomposition}) will have large entries which must nearly cancel when recombined into a solution of (\ref{eq:searchstepsolutionblocktwo}).  Near-cancellation reduces numerical precision.

These problems stem from the off-diagonal blocks of $T$, which ultimately come from the free variables $y$ in our semidefinite program.  Several authors have considered the problem of efficiently and stably solving semidefinite programs with free variables, with no obvious consensus \cite{Anjos:2007:HFV:1350561.1350570}.  For example, \cite{Kobayashi:2007:CSF:1236567.1236578} suggests eliminating free variables by explicitly solving the primal constraint $B^T x = b$ and taking appropriate linear combinations of the matrices $A_p$.  However this procedure destroys the sparsity structure of $S$, making it no longer block diagonal and forcing us to use an expensive full Cholesky decomposition.

Preserving the structure of $S$ and $T$ is paramount.  The simplest way to stabilize (\ref{eq:naiveblockdecomposition}) is to increase the precision of the underlying arithmetic.  In practice, this appears to be necessary anyway for larger-scale bootstrap problems, see appendix~\ref{sec:parameters}.  For additional stabilization, we use an old trick of adding low-rank pieces to $S$ to make it better-conditioned.  Suppose $u_i$ are vectors, each with a single nonzero entry, such that $S'\equiv S+\sum_i u_i u_i^T= S+U U^T$ has no small eigenvalues.\footnote{\label{footnote:stabilize}The $u_i$ can be found as follows. During Cholesky decomposition $S=L L^T$, we keep track of the geometric mean $\Lambda$ of the diagonal entries. Whenever we encounter a diagonal entry $L_{ii}<\theta \Lambda$, where $\theta \ll 1$ is a parameter, we replace $L_{ii}\to L_{ii}+\Lambda$, which amounts to choosing $u_i=\Lambda e_i$ where $e_i$ is a unit vector in the $i$-th direction.}  Note that $S'$ differs from $S$ in only a few diagonal entries --- in particular it has the same block structure.  Now let us replace (\ref{eq:searchstepsolutionblocktwo}) with the system
\be
\label{eq:stabilizedschurcomplement}
T'\begin{pmatrix}
dx \\
dy \\
z
\end{pmatrix}
=
\begin{pmatrix}
S' & -B & -U \\
B^T & 0 & 0\\
U^T & 0 & -\mathbf{1}
\end{pmatrix}
\begin{pmatrix}
dx \\
dy \\
z
\end{pmatrix}
=
\begin{pmatrix}
-d-\Tr(A_* Z)\\
p\\
0
\end{pmatrix}.
\ee
By solving for $z$ and substituting back in, it is easy to see that (\ref{eq:stabilizedschurcomplement}) is precisely equivalent to (\ref{eq:searchstepsolutionblocktwo}).  However, the advantage is that because $S'$ is well-conditioned, a decomposition like (\ref{eq:naiveblockdecomposition}) is numerically stable.  Indeed, defining $B'=(B\ U)$, we have
\be
\label{eq:modifieddecomposition}
T'
&=&
\begin{pmatrix}
L' & 0\\
B'^T L'^{-T} & \mathbf{1}
\end{pmatrix}
\begin{pmatrix}
\mathbf{1} & 0\\
0 & Q'
\end{pmatrix}
\begin{pmatrix}
L'^T & -L'^{-1}B'\\
0 & \mathbf{1}
\end{pmatrix},\\
Q' &=& B'^T L'^{-T} L'^{-1} B' - \begin{pmatrix}
0 & 0\\
0 & \mathbf{1}
\end{pmatrix},
\ee
where $S'=L'L'^T$.  Because $Q'$ is no longer necessarily positive semidefinite, we are forced to use an LU decomposition to invert the middle matrix in (\ref{eq:modifieddecomposition}), which is slightly more expensive than Cholesky decomposition.  Fortunately, $Q'$ is usually small enough that this cost is inconsequential. If $T$ itself is ill-conditioned, then this will manifest as instabilities when we try to LU decompose $Q'$.  In this situation, there is little we can do to avoid imprecision.

\subsubsection{Computing the Schur Complement Matrix}
\label{sec:computingschur}

Now that we know what to do with the Schur complement matrix $S$, let us compute it efficiently.  
The fact that $A_p$ has low rank is helpful.  This explains why we chose to evaluate the polynomial equalities (\ref{eq:polynomialequality}) at discrete points $x_k$ in (\ref{eq:linearPMPequalities}).  Matching polynomial coefficients on each side, as done in \cite{Poland:2011ey,Kos:2013tga,Kos:2014bka}, leads to higher-rank matrices $A_p$.  The helpfulness of low-rank constraints in solving SDPs, and their appearance in polynomial optimization, is well known \cite{4434343}.

Recall that $X$ and $Y$ are block diagonal (\ref{eq:blockstructureforY}), with each block $X_b$ acting on a tensor product of the form $\R^{\delta+1}\otimes \R^m$.  Let $X_b^{(r,s)} \in \R^{(\delta+1)\x (\delta+1)}$ denote the $(r,s)$-th block of $X_b$ in the second tensor factor, which acts on $\R^{\delta+1}$.

Since $S$ is block diagonal, we need only compute elements with $j_1=j_2=j$.  We have
\be
\label{eq:schurcomplementusingpairings}
S_{(j,r_1,s_1,k_1),(j,r_2,s_2,k_2)} &=& \sum_{b\,\in\,\texttt{blocks}_j}\Tr((\vec v_{b,k_1}\vec v_{b,k_1}^T\otimes E^{r_1s_1})X_b^{-1}(\vec v_{b,k_2}\vec v_{b,k_2}^T\otimes E^{r_2s_2})Y_b)\nn\\
&=& \sum_{b\,\in\,\texttt{blocks}_j} \frac 1 4 \Big((\vec v_{b,k_1}^T (X^{-1}_b)^{(s_1,r_2)} \vec v_{b,k_2})(\vec v_{b,k_2}^T Y_b^{(s_2,r_1)}\vec v_{b,k_1}) \nn\\
&& \qquad\qquad+\, (r_1\leftrightarrow s_1) + (r_2\leftrightarrow s_2)+ (r_1\leftrightarrow s_1,r_2\leftrightarrow s_2)\Big).
\label{eq:schurintermsofbilinearpairings}
\ee
Thus, instead of performing repeated matrix multiplications to calculate $\Tr(A_{p_1} X^{-1} A_{p_2} Y)$, we can precompute the bilinear pairings
\be
\label{eq:bilinearpairings}
U^{(b)}_{(d_j+1) s+k_1, (d_j+1) r + k_2} &\equiv& \vec v_{b,k_1}^T (X^{-1}_b)^{(s,r)} \vec v_{b,k_2},\\
V^{(b)}_{(d_j+1) s + k_2, (d_j+1) r + k_1} &\equiv & \vec v_{b,k_2}^T Y_b^{(s,r)}\vec v_{b,k_1},
\ee
and plug them into (\ref{eq:schurintermsofbilinearpairings}).  Whereas forming $S$ is often the most expensive operation in less-specialized solvers, the method outlined here makes it subdominant to other computations, see appendix~\ref{sec:complexity}.

\subsubsection{Computing Step Lengths}

To find step lengths $\alpha_\cP, \alpha_\cD$, we must be able to compute $\a(M,dM)$ where $M$ is a positive semidefinite matrix and $\a(M,dM)$ is the largest positive real number such that $M+\a(M,dM) dM$ is positive semidefinite.  Let $M=LL^T$ be a Cholesky decomposition of $M$.  Since $M+\a dM=L(\mathbf{1}+\a L^{-1}dM L^{-T})L^T$, we have
\be
\alpha(M,dM) &=& -1/\textrm{min-eigenvalue}(L^{-1} dM L^{-T}).
\ee
In \SDPB, we compute all the eigenvalues of $L^{-1} dM L^{-T}$ using a QR decomposition and then simply take the minimum.  Some solvers, like \texttt{SDPA}, implement the more efficient Lanczos method \cite{Lanczos50aniterative} for computing the minimum eigenvalue.  In practice, the step-length calculation is a small part of the total running time, so we have neglected this optimization.

\subsection{Implementation}

\SDPB\ is approximately 3500 lines of C++.  It uses the GNU Multiprecision Library (GMP) \cite{Granlund12} for arbitrary precision arithmetic, and MPACK \cite{MPACK} for arbitrary precision linear algebra.  The relevant MPACK source files are included with \SDPB, with some slight modifications:
\begin{itemize}
\item The Cholesky decomposition routine {\tt Rpotrf} has been modified to implement the stabilization procedure described in footnote~\ref{footnote:stabilize}.

\item Some loops in the LU decomposition routine {\tt Rgetrf} have been parallelized.

\end{itemize}
\SDPB\ also depends on the Boost C++ libraries \cite{BoostSite} and the parsing library \texttt{tinyxml2} \cite{TINYXML2}.

Previous experience shows that high-precision arithmetic is important for accurately solving bootstrap optimization problems.  It is not fully understood why.  The na\"ive reason is that derivatives $\ptl_z^m \ptl_{\bar z}^n g_{\Delta,\ell}(z,\bar z)$  of conformal blocks vary by many orders of magnitude relative to each other as $\Delta$ varies.  It is not possible to scale away this large variation, and answers may depend on near cancellation of large numbers.  In practice, the matrix manipulations in our interior point algorithm ``leak" precision, so that the search direction $(dx,dX,dy,dY)$ is less precise than the initial point $(x,X,y,Y)$.  By increasing the precision of the underlying arithmetic, the search direction can be made reliable again.  This strategy (which we adopt) comes at a cost of increased runtime and memory usage.  Better strategies for dealing with numerical instabilities in bootstrap problems could bring enormous gains.

\SDPB\ is parallelized with OpenMP \cite{openmp08}.  Because most matrices appearing in the interior point algorithm are block-diagonal, most computations are ``embarrassingly parallel:" different blocks can be distributed to different threads.  (The most prominent exception is the LU decomposition of $Q'$, which is why \texttt{Rgetrf} was modified.)  Consequently, performance scales nearly linearly with the number of cores, as long as the number of matrix blocks is sufficiently large.  This is usually the case for interesting bootstrap problems, where $J$ (which sets the number of blocks) is typically much larger than the number of available cores.  It should be possible to achieve favorable scaling up to dozens or even hundreds of cores using MPI and more careful memory management.  Further scaling should be possible with more fine-grained parallelization.

\SDPB\ is available online at \href{https://github.com/davidsd/sdpb}{\texttt{https://github.com/davidsd/sdpb}} under the MIT license.  The source code is carefully commented and written for readability (to the extent that C++ code is ever readable).  We hope this will encourage customization and improvement.

\section{Example Application: 3d Ising Critical Exponents}
\label{sec:applicationtobootstrap}

\subsection{A 3d Ising Optimization Problem}

Bootstrap optimization problems can be naturally approximated as PMPs \cite{Poland:2011ey,Kos:2014bka}.  In this section, we review the PMP for the system of correlators $\{\<\s\s\s\s\>,\<\s\s\e\e\>,\<\e\e\e\e\>\}$ in the 3d Ising CFT.  We will be brief.  Much more detail is given in \cite{Kos:2014bka}.

Associativity of the Operator Product Expansion (OPE) for $\{\<\s\s\s\s\>,\<\s\s\e\e\>,\<\e\e\e\e\>\}$ implies the consistency condition
\be
\label{eq:crossingequationwithv}
\begin{pmatrix}1 & 1\end{pmatrix} \vec{V}_{+,0,0}\begin{pmatrix} 1 \\ 1 \end{pmatrix} + \sum_{\cO^+} \begin{pmatrix}\l_{\s\s\cO} & \l_{\e\e\cO}\end{pmatrix} \vec{V}_{+,\De,\ell}\begin{pmatrix} \l_{\s\s\cO} \\ \l_{\e\e\cO} \end{pmatrix}+ \sum_{\cO^-} \l_{\s\e\cO}^2 \vec{V}_{-,\De,\ell} & = & 0.
\ee
Here, $\cO^+$ runs over $\Z_2$-even operators of even spin and $\cO^-$ runs over $\Z_2$-odd operators of any spin.  We have separated out the unit operator.  $\De$ and $\ell$ are the dimension and spin of $\cO$, respectively.  The object $\vec{V}_{-,\De,\ell}$ is a 5-vector and $\vec{V}_{+,\De,\ell}$ is a 5-vector of $2 \times 2$ matrices
\bea
\vec{V}_{-,\De,\ell} = \begin{pmatrix} 0  \\ 0 \\ F_{-,\De,\ell}^{\s\e,\s\e}(u,v) \\ (-1)^{\ell} F_{-,\De,\ell}^{\e\s,\s\e}(u,v) \\ - (-1)^{\ell} F_{+,\De,\ell}^{\e\s,\s\e}(u,v) \end{pmatrix}, && 
\vec{V}_{+,\De,\ell} = \begin{pmatrix} \begin{pmatrix}  F^{\s\s,\s\s}_{-,\De,\ell}(u,v) & 0 \\ 0 & 0  \end{pmatrix} \\ \begin{pmatrix}  0 & 0 \\ 0 & F^{\e\e,\e\e}_{-,\De,\ell}(u,v)  \end{pmatrix}\\ \begin{pmatrix}  0 & 0 \\ 0 & 0  \end{pmatrix}  \\ \begin{pmatrix}  0 & \frac12 F^{\s\s,\e\e}_{-,\De,\ell}(u,v) \\ \frac12 F^{\s\s,\e\e}_{-,\De,\ell}(u,v) & 0 \end{pmatrix} \\ \begin{pmatrix} 0 & \frac12 F^{\s\s,\e\e}_{+,\De,\ell}(u,v) \\ \frac12 F^{\s\s,\e\e}_{+,\De,\ell}(u,v) & 0  \end{pmatrix} \end{pmatrix},\qquad
\eea
with entries that are functions of conformal cross-ratios $u$ and $v$,
\be
F_{\pm,\De,\ell}^{ij,kl}(u,v)  &\equiv& v^{\frac{\De_k+\De_j}{2}} g_{\De,\ell}^{\De_{ij},\De_{kl}}(u,v) \pm  u^{\frac{\De_k+\De_j}{2}} g_{\De,\ell}^{\De_{ij},\De_{kl}}(v,u),\nn\\
\De_{ij} &\equiv& \De_i-\De_j.
\ee
The $g_{\De,\ell}^{\De_{ij},\De_{kl}}(v,u)$ are conformal blocks, which are known special functions.

The OPE coefficients $\l_{\s\s\cO},\l_{\s\e\cO}, \l_{\e\e\cO}$ and dimensions $\De$ are not known a priori.  Nonetheless, we can constrain them by understanding when it is possible for the terms in (\ref{eq:crossingequationwithv}) to sum to zero.  To do this, consider a 5-vector of functionals $\vec \alpha =(\alpha^1, \dots,\alpha^5)$, where each $\a^i$ acts on the space of functions of $u$ and $v$.  Acting on (\ref{eq:crossingequationwithv}) with $\vec\a$ gives
\be
\label{eq:dualcrossingequationwithv}
\begin{pmatrix} 1 & 1\end{pmatrix} \vec \alpha \cdot \vec{V}_{+,0,0} \begin{pmatrix} 1 \\ 1 \end{pmatrix} + \sum_{\cO^+} \begin{pmatrix} \l_{\s\s\cO} & \l_{\e\e\cO}\end{pmatrix} \vec \alpha \cdot \vec{V}_{+,\De,\ell} \begin{pmatrix} \l_{\s\s\cO} \\ \l_{\e\e\cO} \end{pmatrix} + \sum_{\cO^-} \l_{\s\e\cO}^2 \vec \alpha \cdot \vec{V}_{-,\De,\ell} = 0 . \nn\\
\ee

The bootstrap logic, in the spirit of \cite{Rattazzi:2008pe}, is as follows.  First we make an assumption about which $\De,\ell$ appear in (\ref{eq:dualcrossingequationwithv}).  We then search for a functional $\vec\a$ such that
\bea
\label{eq:functionalinequalities}
\begin{pmatrix} 1 & 1\end{pmatrix} \vec \alpha \cdot \vec{V}_{+,0,0} \begin{pmatrix} 1 \\ 1 \end{pmatrix}  &>& 0 , \nn\\
\vec \alpha \cdot \vec{V}_{+,\De,\ell} &\succeq &0 ,\quad\textrm{for all $\mathbb{Z}_2$-even operators with even spin,} \nn\\
\vec \alpha \cdot \vec{V}_{-,\De,\ell} & \ge &0 ,\quad \textrm{for all $\mathbb{Z}_2$-odd operators of any spin.} 
\eea
The OPE coefficients $\l_{ijk}$ are real in a unitary CFT.  Thus, if $\vec\a$ exists, then it is impossible to satisfy the consistency condition (\ref{eq:crossingequationwithv}), and our assumption is ruled out.  By making different assumptions and searching for functionals $\vec\a$, we can map out the space of allowed $\De$.

\subsection{Approximation as a PMP}

The conditions (\ref{eq:functionalinequalities}) define a feasibility problem with an infinite number of semidefiniteness constraints (one for each $\De,\ell$).  To obtain a PMP, we choose a particular type of functional
\be
\label{eq:alphaintermsofderivatives}
\alpha^i [ f ] = \sum_{\substack{m \geq n \\ m+n\leq\Lambda}} \left. a^i_{mn}\partial_z^m \partial_{\bar z}^n f(z, \bar z)\right|_{z = \bar z = \frac 1 2},
\ee
where $u=z\bar z$, $v=(1-z)(1-\bar z)$.  Although the bootstrap logic does not depend on the types of functionals considered, only functionals of the form (\ref{eq:alphaintermsofderivatives}) lead to a PMP.  Other types of functionals require different optimization methods.

Derivatives of conformal blocks have a systematic approximation in terms of positive functions times polynomials,
\be
\left.\partial_z^m \partial_{\bar z}^n g_{\De,\ell}^{\De_{12},\De_{34}}(z, \bar z)\right|_{z = \bar z = \frac 1 2}\aeq \chi_{\ell}(\De) p_\ell^{\De_{12},\De_{34};mn}(\Delta),
\label{eq:polynomialapproximationforblock}
\ee
where $p_\ell^{\De_{12},\De_{34};mn}(\De)$ are polynomials and $\chi_{\ell}(\De)$ are functions that are positive for all $\De$ in a unitary CFT. It follows that
\be
\left. \partial_z^m \partial_{\bar z}^n F_{\pm,\De,\ell}^{ij,kl}(z, \bar z)\right|_{z = \bar z = \frac 1 2} \aeq \chi_{\ell}(\De) P_{\pm,\ell}^{ij,kl;mn}(\De),
\ee
where $P_{\pm,\ell}^{ijkl;mn}(\De)$ are linear combinations of $p_{\ell}^{\De_{ij},\De_{kl};mn}(\De)$.  Using this approximation, and stripping off the positive factors $\chi_\ell(\De)$, (\ref{eq:functionalinequalities}) becomes a PMP:
\bea
\textrm{find $a_{mn}^i$ such that:}&&\nn\\
\begin{pmatrix} 1 & 1\end{pmatrix}Z_0(0) \begin{pmatrix} 1 \\ 1 \end{pmatrix}  &>& 0 , \nn\\
Z_{\ell}(\De) &\succeq &0 ,\quad\textrm{for all $\mathbb{Z}_2$-even operators with even spin,} \nn\\
Y_{\ell}(\De) & \ge &0 ,\quad \textrm{for all $\mathbb{Z}_2$-odd operators in the spectrum.} 
\label{eq:constraintsformixedcorrelator}
\eea
Here $Y_\ell(\De)$ are polynomials and $Z_\ell(\De)$ are polynomial matrices in $\De$ defined as
\be
Y_{\ell}(\Delta)  &\equiv& \sum_{mn} \left[ a^3_{mn} P^{\s\e,\s\e;mn}_{-,\ell}(\De) + a^4_{mn} (-1)^{\ell} P^{\e\s,\s\e,mn}_{-,\ell}(\De) - a^5_{mn} (-1)^{\ell} P^{\e\s,\s\e;mn}_{+, \ell}(\De) \right],  \notag\\
Z_{\ell}(\Delta) &\equiv&  \sum_{mn} \begin{pmatrix}a^1_{mn} P^{\s\s,\s\s;mn}_{-,\ell}(\De) & \frac12 \left( a^4_{mn} P_{-,\ell}^{\s\s,\e\e;mn}(\De) + a^5_{mn} P_{+,\ell}^{\s\s,\e\e;mn}(\De) \right)   \\ \frac12 \left( a^4_{mn} P_{-,\ell}^{\s\s,\e\e;mn}(\De) + a^5_{mn} P_{+,\ell}^{\s\s,\e\e;mn}(\De) \right) & a^2_{mn} P^{\e\e,\e\e;mn}_{-,\ell}(\De)\end{pmatrix} \notag.\\
\ee
Typically, we assume that $\De$ can vary arbitrarily above some minimum value $\De_\mathrm{min}(\ell)$.  Writing $\De=\De_\mathrm{min}(\ell)+x$, we have positive semidefiniteness for all $x\geq 0$.

There are two important differences between (\ref{eq:constraintsformixedcorrelator}) and our original PMP (\ref{eq:PMPconstraint}):
\begin{enumerate}
\item In (\ref{eq:PMPconstraint}) we have a finite number of positive semidefiniteness conditions $j=1,\dots,J$, whereas here we have an infinite number since $\ell$ can be any nonnegative integer.  In practice, we include spins $\ell$ up to some large but finite $\ell_{\max}$.  As long as $\ell_{\max}$ is large enough, a functional obtained by solving the problem with $\ell\leq \ell_{\max}$ should also satisfy positive semidefiniteness for spins $\ell>\ell_{\max}$.  The proper choice of $\ell_{\max}$ depends on the problem at hand, see appendix~\ref{sec:parameters}.  See \cite{Caracciolo:2014cxa} for a more careful analysis.

\item In (\ref{eq:PMPconstraint}) we are trying to optimize an objective function $b\. y$, whereas here we are only interested in feasibility.  To determine feasibility, we can pick the trivial objective function $b=0$ and run our interior point algorithm until it becomes dual feasible.

\end{enumerate}

\subsection{Setting Up \SDPB}

The natural objects entering our calculation are (approximate) derivatives of conformal blocks $\chi_{\ell}(\De) p_\ell^{\De_{12},\De_{34};mn}(\Delta)$, as opposed to just the polynomials $p_\ell^{\De_{12},\De_{34};mn}(\Delta)$.  Removing positive factors does not affect positive semidefiniteness, but it does affect the scaling of the resulting SDP.  Restoring quantities to their ``natural" size can improve numerical stability and performance.\footnote{A similar observation was made for the algorithm in \cite{El-Showk:2014dwa}.}  \SDPB\ provides a few different ways to implement this rescaling.

As we saw in section~\ref{sec:translationPMPtoSDP}, translating a PMP into an SDP requires a bilinear basis $q_m(x)$.  \SDPB\ allows a choice of bilinear basis $q_m^{(j)}(x)$ for each $j=1,\dots,J$.  For bootstrap problems, we take $q_m^{(j)}(x)$ to be orthogonal polynomials with respect to the norm
\be
\label{eq:candidatenorm}
\<p,q\>^{(j)} &=& \int_{0}^\oo dx\,\chi_\ell(\De_{\min}(\ell)+x) p(x) q(x),
\ee
where $\ell$ is the spin corresponding to $j$. The change of basis between orthogonal polynomials with respect to $\<\.,\.\>^{(j)}$ and monomials $x^m$ (used in previous bootstrap applications of SDP) is extremely ill-conditioned at high degree.  So although the choice of $q_m^{(j)}$ is unimportant in principle, it can have a dramatic effect on numerical stability.\footnote{We thank Pablo Parrilo for pointing out the usefulness of orthogonal polynomials in improving the numerical stability of polynomial optimization.}

\SDPB\ also requires a set of sample points $x^{(j)}_k$ at which to evaluate polynomials, as well as scaling factors $s^{(j)}_k$ that modify the constraints (\ref{eq:BforPMP}, \ref{eq:CforPMP}) as follows:
\be
B_{(j,r,s,k),n} &=& -s^{(j)}_k P_{j,rs}^{n}(x^{(j)}_k),\nn\\
c_{(j,r,s,k)} &=& s^{(j)}_k P^{0}_{j,rs}(x^{(j)}_k).
\ee
Additionally, the $A_{(j,r,s,k)}$ are given by (\ref{eq:specialconstraintmatrices}) with
\be
v_{2j-1,k} &=&  (s_k^{(j)})^{1/2} \vec q^{(j)}_{\de 1}(x^{(j)}_k),\nn\\
v_{2j,k} &=& (s_k^{(j)})^{1/2}(x^{(j)}_k)^{1/2} \vec q^{(j)}_{\de 2}(x^{(j)}_k).
\ee
This $s_k$-dependent rescaling gives an isomorphic, but potentially more numerically stable SDP.  For bootstrap problems, it is natural to pick $s_k^{(j)}=\chi_\ell(x_k^{(j)})$ where $\ell$ corresponds to $j$.  The $x_k^{(j)}$ can be any sequence of distinct points.  A natural choice are zeros of one of the $q_m^{(j)}$ of sufficiently high degree.

To summarize, \SDPB\ depends on the following input:
\begin{itemize}
\item for each $j=1,\dots,J$:
\begin{itemize}
\item polynomial matrices $M^0_j(x),\dots,M^N_j(x)$ of maximum degree $d_j$,
\item bilinear bases $q_m^{(j)}(x)$ ($m=0,\dots,\lfloor d_j/2\rfloor$),
\item sample points $x_k^{(j)}$ ($k=0,\dots,d_j$),
\item sample scalings $s_k^{(j)}$ ($k=0,\dots,d_j$),
\end{itemize}
\item an objective function $b\in \R^N$.
\end{itemize}

\SDPB\ reads this data in an {\tt XML} format described in the manual.  A {\tt Mathematica} package  that translates the above data from {\tt Mathematica} expressions into {\tt XML} is included with the source distribution.  An example 2d bootstrap computation is also included.  More details about \SDPB's input and output formats and its various settings can be found in the manual.

\subsection{Results}
\label{sec:example3dIsing}

As an application of \SDPB, let us improve upon the determinations of 3d Ising critical exponents in \cite{El-Showk:2014dwa,Kos:2014bka}.  We make an exclusion plot for the operator dimensions $(\Delta_\s,\Delta_\e)$ as follows.  Fix $(\Delta_\s,\Delta_\e)$ and use \SDPB\ to determine if the PMP (\ref{eq:constraintsformixedcorrelator}) is dual-feasible, i.e. whether there exist $(y,Y)$ satisfying their associated constraints.  If the PMP is dual-feasible, then the given $(\Delta_\s,\Delta_\e)$ are excluded.  If the PMP is not dual-feasible, then we cannot conclude anything about $(\Delta_\s,\Delta_\e)$.  By scanning over different points, we map out the excluded region in $(\Delta_\s,\Delta_\e)$ space.

To determine dual feasibility, we use a vanishing objective function $b=0$ and run \SDPB\ with the option \texttt{--findDualFeasible}.  This terminates the solver if $(y,Y)$ are found satisfying their constraints to sufficient precision (i.e. if $\texttt{dualError}<\texttt{dualErrorThreshold}$).\footnote{If we kept running the solver, \texttt{primalObjective} would converge towards $\texttt{dualObjective}=0$ and an optimum would eventually be reached.}  In practice, if \SDPB\ finds a primal feasible solution $(x,X)$ after some number of iterations, then it will never eventually find a dual feasible one.  Thus, we additionally include the option \texttt{--findPrimalFeasible}, which terminates the solver whenever $\texttt{primalError}<\texttt{primalErrorThreshold}$.  The termination status of \SDPB\ then determines whether a point is allowed or not:
\be
\begin{array}{rcl}
\texttt{found dual feasible solution} & \implies & \textrm{$(\Delta_\s,\Delta_\e)$ disallowed},\\
\texttt{found primal feasible solution} & \implies & \textrm{$(\Delta_\s,\Delta_\e)$ allowed}.
\end{array}
\ee
The precise \SDPB\ options used for the computations in this work are described in appendix~\ref{sec:parameters}.

\begin{figure}[t!]
\begin{center}
\begin{psfrags}
\def\PFGstripminus-#1{#1}%
\def\PFGshift(#1,#2)#3{\raisebox{#2}[\height][\depth]{\hbox{%
  \ifdim#1<0pt\kern#1 #3\kern\PFGstripminus#1\else\kern#1 #3\kern-#1\fi}}}%
\providecommand{\PFGstyle}{}%
\psfrag{title}[cc][cc]{\PFGstyle $\text{allowed region for $\Lambda=19,27,35,43$}$}
\psfrag{xLabel}[cl][cl]{\PFGstyle $\text{$\De_\s$}$}
\psfrag{yLabel}[bc][bc]{\PFGstyle $\text{$\De_\e$}$}
\psfrag{x051805}[tc][tc]{\PFGstyle $0.51805$}
\psfrag{x05181}[tc][tc]{\PFGstyle $0.5181$}
\psfrag{x051815}[tc][tc]{\PFGstyle $0.51815$}
\psfrag{x05182}[tc][tc]{\PFGstyle $0.5182$}
\psfrag{x051825}[tc][tc]{\PFGstyle $0.51825$}
\psfrag{x05183}[tc][tc]{\PFGstyle $0.5183$}
\psfrag{x051835}[tc][tc]{\PFGstyle $0.51835$}
\psfrag{y14115}[cr][cr]{\PFGstyle $1.4115$}
\psfrag{y1412}[cr][cr]{\PFGstyle $1.412$}
\psfrag{y14125}[cr][cr]{\PFGstyle $1.4125$}
\psfrag{y1413}[cr][cr]{\PFGstyle $1.413$}
\psfrag{y14135}[cr][cr]{\PFGstyle $1.4135$}
\psfrag{y1414}[cr][cr]{\PFGstyle $1.414$}
\psfrag{y14145}[cr][cr]{\PFGstyle $1.4145$}
\includegraphics[width=0.95\textwidth]{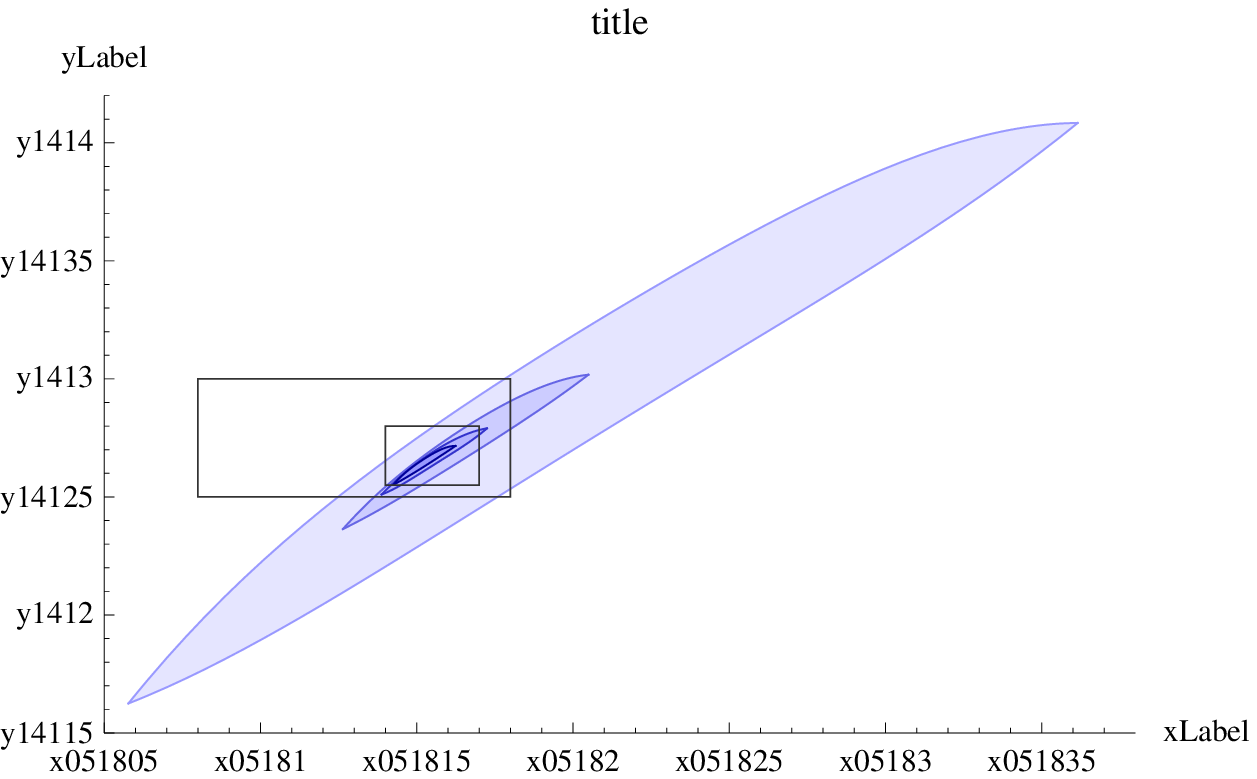}
\end{psfrags}
\caption{Allowed region for a $\Z_2$-symmetric 3d CFT with two relevant scalars, computed using \SDPB\ with the system of correlators $\<\s\s\s\s\>,\<\s\s\e\e\>,$ and $\<\e\e\e\e\>$.  The blue regions correspond to $\Lambda=19,27,35,43$, in decreasing order of size.  The larger black rectangle shows the current most precise Monte Carlo determinations of critical exponents in the 3d Ising CFT \cite{Hasenbusch:2011yya}.  The smaller black rectangle shows the estimate for $(\Delta_\s,\Delta_\e)$ using $c$-minimization at $\Lambda=41$ for the single correlator $\<\s\s\s\s\>$ \cite{El-Showk:2014dwa}.}
\label{fig:multicorrelatorRegionDifferentNmax}
\end{center}
\end{figure}

In figure~\ref{fig:multicorrelatorRegionDifferentNmax}, we plot the allowed regions for different numbers of derivatives labeled by $\Lambda=19,27,35,43$,\footnote{$\nmax=10,14,18,22$ in the notation of \cite{Kos:2014bka}.} corresponding to functionals $\vec\alpha$ of dimension 275, 525, 855, and 1265, respectively.\footnote{A performance analysis for different values of $\Lambda$ is given in appendix~\ref{app:runningtime}.}   We focus on $(\Delta_\s,\Delta_\e)$ near the 3d Ising CFT, leaving wider exploration to the future.  The allowed region is an island that shrinks rapidly with increasing $\L$.\footnote{Each allowed region plotted in this work was computed by testing a grid of points and fitting curves to the boundary between allowed and disallowed gridpoints.  The raw gridpoint data is available on request.}  The largest island, corresponding to $\Lambda=19$ is the same as the allowed region in figure~5 of \cite{Kos:2014bka}.  We can estimate the point towards which the islands shrink as follows.  Let $(a_\L,b_\L)$ be the bottom-left point of the $\L$-allowed island, and similarly let $(c_\L,d_\L)$ be the top-right point.  Define 
\be
E_x(r)&=&\mathrm{stddev}_{\L\in\{19,27,35,43\}}(ra_\L+(1-r)c_\L),\nn\\
E_y(r)&=&\mathrm{stddev}_{\L\in\{19,27,35,43\}}(rb_\L+(1-r)d_\L),
\ee
and let $r_x,r_y$ be the minima of $E_x, E_y$ respectively.  Our estimate is\footnote{For readers interested in numerology, we recommend \cite{ISC}.  However, see \cite{PhysRev.82.554.2}.}
\be
(\De_\s,\De_\e)&\approx& (r_x a_{43}+(1-r_x)c_{43}, r_y b_{43} + (1-r_y) d_{43})\nn\\
&\approx& (0.5181478(5), 1.412617(4)),
\ee
where the errors are given by $E_x(r_x), E_y(r_y)$.

The dimensions $(\Delta_\s,\Delta_\e)$ were estimated in \cite{El-Showk:2014dwa} using the conjecture that the 3d Ising CFT minimizes $c \equiv \Delta_\s^2/\lambda_{\s\s T_{\mu\nu}}$ subject to the constraints of unitarity and crossing symmetry of $\<\s\s\s\s\>$.  This conjecture, called ``$c$-minimization," is expected to be equivalent to the assumption that the 3d Ising CFT lives precisely at the kink in the dimension bound coming from the single correlator $\<\s\s\s\s\>$.  Although unproven, $c$-minimization's advantage is that it allows one to estimate $(\Delta_\s,\Delta_\e)$ using a single scalar correlator $\<\s\s\s\s\>$. Bootstrap computations for a single scalar correlator can be made relatively efficient with a modified primal simplex algorithm \cite{El-Showk:2014dwa,Paulos:2014vya}.

An advantage of multiple correlators is that it is possible to impose the condition that $\s$ is the only relevant $\Z_2$-odd operator, causing the allowed region in $(\Delta_\s,\Delta_\e)$-space to become a closed island, independent of auxiliary assumptions \cite{Kos:2014bka}.  Because our $\Lambda=43$ island lies within the error bars of \cite{El-Showk:2014dwa}, our results verify $c$-minimization to the precision achieved in \cite{El-Showk:2014dwa}.\footnote{To see this more explicitly, it should be possible to place both upper and lower bounds on $c$ using \SDPB\ and see that it is constrained to lie close to the minimum computed in \cite{El-Showk:2014dwa}.  We leave this to future work.}  Our results also give further evidence for the conjecture that the 3d Ising CFT is the unique $\Z_2$-symmetric 3d CFT with two relevant scalars and $\Delta_\s \lesssim 0.6$. (The precise condition on $\Delta_\s,\Delta_\e$ depends on the shape of the allowed region further away from the 3d Ising point.)  It would be very interesting to prove these conjectures analytically, perhaps by showing that the island in figure~\ref{fig:multicorrelatorRegionZoom} shrinks to a point as $\Lambda\to\oo$.  An alternative is that the conjectures are still true, but one needs information from other four-point functions to prove them.

\begin{figure}[t!]
\begin{center}
\begin{psfrags}
\def\PFGstripminus-#1{#1}%
\def\PFGshift(#1,#2)#3{\raisebox{#2}[\height][\depth]{\hbox{%
  \ifdim#1<0pt\kern#1 #3\kern\PFGstripminus#1\else\kern#1 #3\kern-#1\fi}}}%
\providecommand{\PFGstyle}{}%
\psfrag{title}[cc][cc]{\PFGstyle $\text{allowed region for $\Lambda=43$, using three-point symmetry}$}
\psfrag{xLabel}[cl][cl]{\PFGstyle $\text{$\De_\s$}$}
\psfrag{yLabel}[bc][bc]{\PFGstyle $\text{$\De_\e$}$}
\psfrag{x051813}[tc][tc]{\PFGstyle $0.51813$}
\psfrag{x051814}[tc][tc]{\PFGstyle $0.51814$}
\psfrag{x051815}[tc][tc]{\PFGstyle $0.51815$}
\psfrag{x051816}[tc][tc]{\PFGstyle $0.51816$}
\psfrag{x051817}[tc][tc]{\PFGstyle $0.51817$}
\psfrag{y14125}[cr][cr]{\PFGstyle $1.4125$}
\psfrag{y141255}[cr][cr]{\PFGstyle $1.41255$}
\psfrag{y14126}[cr][cr]{\PFGstyle $1.4126$}
\psfrag{y141265}[cr][cr]{\PFGstyle $1.41265$}
\psfrag{y14127}[cr][cr]{\PFGstyle $1.4127$}
\psfrag{y141275}[cr][cr]{\PFGstyle $1.41275$}
\psfrag{y14128}[cr][cr]{\PFGstyle $1.4128$}
\psfrag{y141285}[cr][cr]{\PFGstyle $1.41285$}
\includegraphics[width=0.95\textwidth]{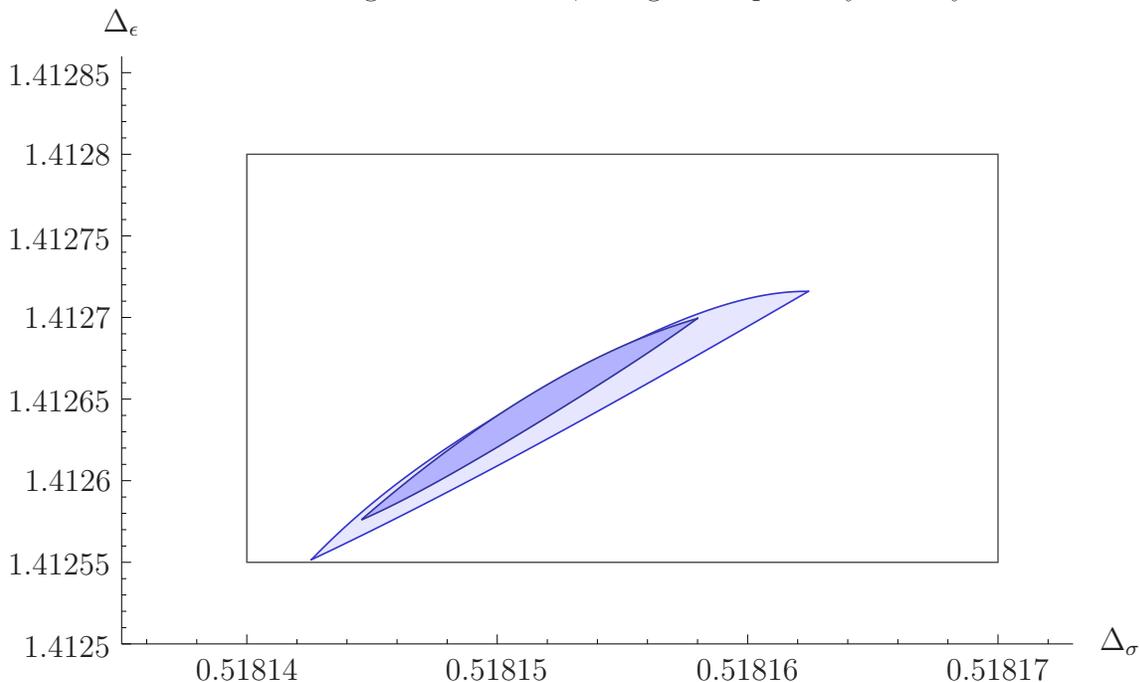}
\end{psfrags}
\caption{Allowed region for a $\Z_2$-symmetric 3d CFT with two relevant operators, computed with \SDPB\ at $\Lambda=43$.  The light-blue region is a zoom of the smallest region in figure~\ref{fig:multicorrelatorRegionDifferentNmax}.  The darker-blue region additionally uses symmetry of the OPE coefficients $\l_{\s\s\e}=\l_{\s\e\s}$.  The black rectangle shows the estimate for $(\Delta_\s,\Delta_\e)$ using $c$-minimization at $\Lambda=41$ \cite{El-Showk:2014dwa}.}
\label{fig:multicorrelatorRegionZoom}
\end{center}
\end{figure}

The analysis of \cite{Kos:2014bka} did not use permutation symmetry of the OPE coefficient $\lambda_{\s\s\e}=\l_{\s\e\s}$.\footnote{Note that permutation symmetry holds only when the conformal blocks are correctly normalized as functions of $\Delta$.  For scalars, the correct normalization is $g_{\Delta,0}(u,v)=u^{\Delta/2}+\dots$ to leading order in $u$, up to a $\Delta$-independent constant.}  Including this constraint leads to an additional modest reduction in the allowed region,\footnote{This fact was discovered during collaboration with Filip Kos, David Poland, and Alessandro Vichi \cite{ONFuture}.} which we plot in figure~\ref{fig:multicorrelatorRegionZoom} at $\Lambda=43$.  The resulting island gives a rigorous determination $\De_\s=0.518151(6)$, $\De_\e=1.41264(6)$, which is 5-10 times more precise than the Monte Carlo results of \cite{Hasenbusch:2011yya}.  We summarize the comparison to Monte Carlo in figure~\ref{multicorrelatorRegionMCComparison}.

\begin{figure}[t!]
\begin{center}
\begin{psfrags}
\def\PFGstripminus-#1{#1}%
\def\PFGshift(#1,#2)#3{\raisebox{#2}[\height][\depth]{\hbox{%
  \ifdim#1<0pt\kern#1 #3\kern\PFGstripminus#1\else\kern#1 #3\kern-#1\fi}}}%
\providecommand{\PFGstyle}{}%
\psfrag{title}[cc][cc]{\PFGstyle $\text{comparison to Monte Carlo}$}
\psfrag{MC}[cc][rc]{\PFGstyle $\text{\ \ \ \ \ \ \ \ Monte Carlo}$}
\psfrag{xLabel}[cl][cl]{\PFGstyle $\text{$\De_\s$}$}
\psfrag{yLabel}[bc][bc]{\PFGstyle $\text{$\De_\e$}$}
\psfrag{x0518}[tc][tc]{\PFGstyle $0.518$}
\psfrag{x051802}[tc][tc]{\PFGstyle $0.51802$}
\psfrag{x051804}[tc][tc]{\PFGstyle $0.51804$}
\psfrag{x051806}[tc][tc]{\PFGstyle $0.51806$}
\psfrag{x051808}[tc][tc]{\PFGstyle $0.51808$}
\psfrag{x05181}[tc][tc]{\PFGstyle $0.5181$}
\psfrag{x051812}[tc][tc]{\PFGstyle $0.51812$}
\psfrag{x051814}[tc][tc]{\PFGstyle $0.51814$}
\psfrag{x051816}[tc][tc]{\PFGstyle $0.51816$}
\psfrag{x051818}[tc][tc]{\PFGstyle $0.51818$}
\psfrag{x05182}[tc][tc]{\PFGstyle $0.5182$}
\psfrag{y1412}[cr][cr]{\PFGstyle $1.412$}
\psfrag{y14121}[cr][cr]{\PFGstyle $1.4121$}
\psfrag{y14122}[cr][cr]{\PFGstyle $1.4122$}
\psfrag{y14123}[cr][cr]{\PFGstyle $1.4123$}
\psfrag{y14124}[cr][cr]{\PFGstyle $1.4124$}
\psfrag{y14125}[cr][cr]{\PFGstyle $1.4125$}
\psfrag{y14126}[cr][cr]{\PFGstyle $1.4126$}
\psfrag{y14127}[cr][cr]{\PFGstyle $1.4127$}
\psfrag{y14128}[cr][cr]{\PFGstyle $1.4128$}
\psfrag{y14129}[cr][cr]{\PFGstyle $1.4129$}
\psfrag{y1413}[cr][cr]{\PFGstyle $1.413$}
\psfrag{y14131}[cr][cr]{\PFGstyle $1.4131$}
\psfrag{y14132}[cr][cr]{\PFGstyle $1.4132$}
\psfrag{y14133}[cr][cr]{\PFGstyle $1.4133$}
\psfrag{y14134}[cr][cr]{\PFGstyle $1.4134$}
\psfrag{y14135}[cr][cr]{\PFGstyle $1.4135$}
\includegraphics[width=0.9\textwidth]{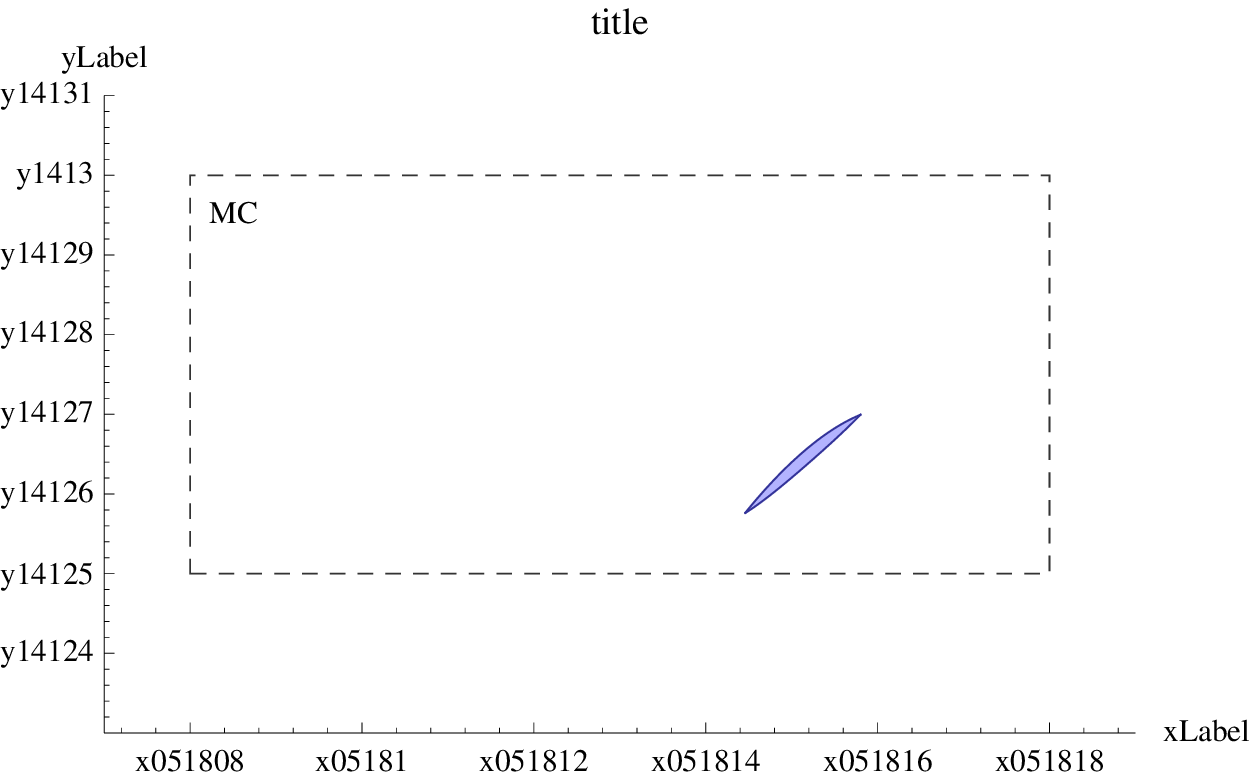}
\end{psfrags}
\caption{Comparison between the allowed region for the 3d Ising CFT using \SDPB\ with $\Lambda=43$ (blue) and Monte Carlo determinations of critical exponents (dashed rectangle) \cite{Hasenbusch:2011yya}.  The size of the Monte Carlo rectangle is set by statistical and systematic errors associated with the simulation.  By contrast, the blue region is a rigorous bound with sharp edges.}
\label{multicorrelatorRegionMCComparison}
\end{center}
\end{figure}

\section{Discussion}
\label{sec:discussion}

With \SDPB, we have significantly improved the precision of $(\De_\s,\De_\e)$ in the 3d Ising CFT.  Our numerics indicate that the window of allowed dimensions may shrink to a point in the limit of infinite computer time.  In other words, they suggest that consistency of the correlators $\{\<\s\s\s\s\>,\<\s\s\e\e\>,\<\e\e\e\e\>\}$, together with the assumption that $\s$ and $\e$ are the only relevant scalars in the theory, may uniquely fix the dimensions $(\De_\s,\De_\e)$. This conjecture could be more tractable analytically than trying to solve the full CFT consistency conditions. 

There are many more 3d Ising observables to explore.  For example, the coefficient $f_{\e\e\e}$ should be computable, and it would be interesting to compare with the recent Monte Carlo prediction \cite{Caselle:2015csa}.  It will also be important to consider larger systems of 3d Ising correlators.

However, \SDPB\ should also enable wider exploration of new correlators and diverse theories.  \SDPB\ is already being used in several bootstrap studies that would have previously been difficult or impossible \cite{ONFuture,FermionFuture}.  An exciting direction that may now be accessible is studying a four-point function of stress-tensors in 3d CFTs.

In addition to the four-point function bootstrap, semidefinite programming has also recently been applied to the ``modular bootstrap" in 2d CFTs \cite{Hellerman:2009bu,Friedan:2013cba}.  \SDPB\ is equally applicable to modular bootstrap computations, since they too can be phrased in terms of polynomial matrix programs.

From the computing point of view, there are many opportunities for improvement.  For example, it should be possible to parallelize \SDPB\ up to hundreds of cores, which could lead to even more precise calculations, and (just as importantly) easier exploration.  Very different algorithms, like Second Order Conic Programming (SOCP), cutting plane methods, or constrained nonlinear optimization may also be applicable.

The revival initiated in \cite{Rattazzi:2008pe} is still young, and the technology (both analytical and numerical) is evolving rapidly.  Current techniques are likely not maximally efficient, and it will be important to consider other methods, from new algorithms and optimization tools to conceptually different approaches.  We are optimistic that much more will be possible.

\section*{Acknowledgements}
I am grateful to Chris Beem, Luca Iliesiu, Silviu Pufu, Slava Rychkov, Balt van Rees, and Ran Yacoby for related discussions, and especially Filip Kos, David Poland, and Alessandro Vichi for discussions, collaboration, assistance testing \SDPB, and comments on the draft.  Thanks also to Slava Rychkov for comments on the draft. Thanks to Amir Ali Ahmadi, Hande Benson, Pablo Parrilo, and Robert Vanderbei for advice on semidefinite programming and numerical optimization.  I am supported by DOE grant number DE-SC0009988 and a William D. Loughlin Membership at the Institute for Advanced Study.
The computations in this paper were run on the Hyperion computing cluster supported by the School of Natural Sciences Computing Staff at the Institute for Advanced Study.

\appendix

\section{Choices and Parameters}
\label{sec:parameters}

The PMP for $\{\<\s\s\s\s\>, \<\s\s\e\e\>, \<\e\e\e\e\>\}$ in the 3d Ising CFT depends on the following choices:
\begin{itemize}
\item An integer $\Lambda$ specifying which derivatives to include in the functional $\vec \alpha$.  These are given by $\ptl_z^m\ptl_{\bar z}^n$ with $m\geq n$ and $m+n\leq \Lambda$.\footnote{$\Lambda=2\nmax-1$, where $\nmax$ is the parameter defined in \cite{Kos:2014bka}.}  Because of the symmetry/antisymmetry of $F_{\pm,\De,\ell}^{ij,kl}(u,v)$ under $u\leftrightarrow v$, some of these derivatives vanish identically.  Including only non-vanishing derivatives, the dimension of $\vec\alpha$ is
\be
\dim \vec\alpha &=& \frac{\lfloor\frac{\Lambda+2}{2}\rfloor(\lfloor\frac{\Lambda+2}{2}\rfloor+1)}{2}+4\frac{\lfloor\frac{\Lambda+1}{2}\rfloor(\lfloor\frac{\Lambda+1}{2}\rfloor+1)}{2}.
\ee

\item An integer $\numax$ controlling the accuracy of the approximation for conformal blocks (\ref{eq:polynomialapproximationforblock}).  The positive prefactor is
\be
\chi_\ell(\Delta) &=& \frac{r_*^\Delta}{\prod_i(\Delta-\Delta_{*i})} \ \ \geq\ \  0,
\ee
where $r_*=3-2\sqrt 2$ is the radius of the point $z=\bar z = \frac 1 2$ in the radial coordinates of \cite{Pappadopulo:2012jk,Hogervorst:2013sma}.  The dimensions $\Delta_{*i}$ are special values {\it below} the unitarity bound, so that the product $\prod_i(\Delta-\Delta_{*i})$ is positive for unitary theories.  The approximation (\ref{eq:polynomialapproximationforblock}) can be systematically improved by including more poles $(\Delta-\Delta_{*i})^{-1}$ and increasing the degree of $p^{\Delta_{ij},\Delta_{kl};mn}_\ell(\Delta)$.
Our choice of poles is
\be
\De_{*i} &\in& \left\{\begin{array}{lcl}
1-\ell-k & | & k=1,\dots,\numax,\\
d/2-k, & | & k=1,\dots,\lfloor\numax/2\rfloor,\\
\ell+d-1-k & | & k=1,\dots,\mathrm{min}(\numax, \lfloor\ell/2\rfloor)
\end{array}\right\}.
\ee

Smaller $\numax$ means smaller-degree polynomials and shorter runtimes.  Larger $\numax$ is needed to get an accurate approximation for conformal blocks.  We choose $\numax$ by computing bounds with successively larger values of $\numax$ until the results stabilize.  Our final values are conservative: smaller $\kappa$ may still give sufficient accuracy.  Derivatives of conformal blocks were computed using the recursion relation in \cite{Kos:2014bka} to order $r^{90}$ (far greater accuracy than needed).

\item A set of spins $S=\{\ell_1,\dots,\ell_L\}$ to include.  If not enough spins are included, the solver may find a functional $\vec\alpha$ that violates a positive semidefiniteness constraint for some spin.  Because derivatives of conformal blocks converge as a function of spin, in practice it is sufficient to include a finite number of spins to ensure $\vec\alpha$ satisfies the constraints for all spins (as can be verified post-hoc by testing $\vec\alpha$ on constraints that were not explicitly included).  This sufficient number of spins grows with $\Lambda$.  Our choices are given in~(\ref{eq:spinsets}).
\be
\label{eq:spinsets}
S_{\L=19} &=& \{0,\dots,26\} \cup \{49,50\}\nn\\
S_{\L=27} &=& \{0,\dots,26\} \cup \{29,30,33,34,37,38,41,42,45,46,49,50\}\nn\\
S_{\L=35} &=& \{0,\dots,44\} \cup \{47,48,51,52,55,56,59,60,63,64,67,68\}\nn\\
S_{\L=43} &=& \{0,\dots,64\} \cup \{67,68,71,72,75,76,79,80,83,84,87,88\}.
\ee

\end{itemize}
Once these quantities are fixed, the parameters to \SDPB\ must be chosen to ensure numerical stability, precision, and correctness.  Our choices for the computations in this work are summarized in table~\ref{tab:parameters}.

\begin{table}[!ht]
\centering
\begin{tabular}{|l|c|c|c|c|}
\hline
$\Lambda$ & 19 & 27 & 35 & 43 \\
$\numax$ & 14 & 20 & 30 & 40 \\
spins & $S_{\L=19}$ & $S_{\L=27}$ & $S_{\L=35}$ & $S_{\L=43}$\\
\texttt{precision} & 448 & 576 & 768 & 960\\
\texttt{findPrimalFeasible} & True & True & True & True\\
\texttt{findDualFeasible} & True & True & True & True\\
\texttt{detectPrimalFeasibleJump} & True & True & True & True\\
\texttt{detectDualFeasibleJump} & True & True & True & True\\
\texttt{dualityGapThreshold} & $10^{-30}$ & $10^{-30}$ & $10^{-30}$ & $10^{-75}$ \\
\texttt{primalErrorThreshold} & $10^{-30}$ & $10^{-30}$ & $10^{-40}$ & $10^{-75}$ \\
\texttt{dualErrorThreshold} & $10^{-30}$ & $10^{-30}$ & $10^{-40}$ & $10^{-75}$ \\
\texttt{initialMatrixScalePrimal} ($\Omega_\cP$) & $10^{40}$ & $10^{50}$ & $10^{50}$ & $10^{60}$\\
\texttt{initialMatrixScaleDual} ($\Omega_\cD$) & $10^{40}$ & $10^{50}$ & $10^{50}$ & $10^{60}$\\
\texttt{feasibleCenteringParameter} ($\beta_\mathrm{feasible}$) & 0.1 & 0.1 & 0.1 & 0.1 \\
\texttt{infeasibleCenteringParameter} ($\beta_\mathrm{infeasible}$) & 0.3 & 0.3 & 0.3 & 0.3\\
\texttt{stepLengthReduction} ($\gamma$) & 0.7 & 0.7 & 0.7 & 0.7\\
\texttt{choleskyStabilizeThreshold} ($\theta$) & $10^{-40}$ & $10^{-40}$ & $10^{-100}$ & $10^{-140}$ \\
\texttt{maxComplementarity} & $10^{100}$ & $10^{130}$ & $10^{160}$ & $10^{200}$\\
\hline
\end{tabular}
\caption{Parameters for the computations in this work.  Only \SDPB\ parameters that affect the numerics (as opposed to parameters like \texttt{maxThreads} and \texttt{maxRuntime}) are included.  The sets of spins $S_{\Lambda}$ are given in (\ref{eq:spinsets}).  Variables in the interior point algorithm of section~\ref{sec:completealgorithm} that correspond to \SDPB\ parameters are indicated in parentheses.  \texttt{precision} is in binary digits.  The spin sets $S_\L$ refer to (\ref{eq:spinsets}).}
\label{tab:parameters}
\end{table}

\section{Performance}

\subsection{Complexity Comparison}
\label{sec:complexity}

Let us compare the complexity of \SDPB's algorithm to that of \SDPAGMP\ for solving PMPs.  For simplicity, suppose that each polynomial matrix $M_j^n(x)$ has size $m\x m$ and degree $d$.  The matrices $X,Y$ then have $2J$ blocks, each of dimension $m(d+1)$.  We focus on the most expensive parts of each iteration of the interior point algorithm and count the number of multiplications to leading order in $J$, $m$, and $d$.

For \SDPAGMP, we assume the setup described in \cite{Poland:2011ey,Kos:2014bka}.  There, the free variables $y$ were embedded as $1\x 1$ diagonal blocks in the matrix $Y$. (By contrast, for \SDPB\ they are treated separately.)  The most important contributions to the running time of \SDPAGMP\ are as follows.
\begin{itemize}
\item The Schur complement matrix is dense, so each of its elements must be computed individually.  Computing $X^{-1} A_q Y$ requires $N+2J(md)^3$ multiplications, since it involves a block-diagonal dense matrix multiplication $X^{-1} \x (A_q Y)$ (the linear term in $N$ comes from multiplying $1\x1$ blocks).  This must be repeated $\frac{Jm^2 d}2$ times: once for each $q$.  Now, the matrices $A_p$ are typically sparse, with $O(md)$ entries.  Thus, evaluating all the traces $\Tr(A_p X^{-1} A_q Y)$ requires approximately $(N+2J(md)^3)\frac{Jm^2 d}{2} + \frac 1 2 (\frac{J m^2 d}{2})^2 (md)$ multiplications.  These steps dominate the running time for the computations in \cite{Kos:2014bka}.

\item Because the Schur complement matrix is dense, it must be inverted using a full Cholesky decomposition, which requires $\frac 1 3 (\frac{Jm^2d}{2})^3$ multiplications.
\end{itemize}

For \SDPB, computing $S$ takes negligible time.  The most important steps are in solving the Schur complement equation:
\begin{itemize}
\item Computing the Cholesky decomposition $S=L L^T$ takes $\frac 1 3 J (\frac{m^2 d}{2})^3$ multiplications, since it can be done block-wise.
\item Forming $L^{-1} B$ requires $NJ\frac {m^2 d}{2}$ multiplications.
\item Forming $Q=(L^{-1} B)^T (L^{-1} B)$ requires $\frac 1 2 N^2 (J \frac{m^2 d}{2})$ multiplications.  This step dominates the running time for the computations in this work.
\item Computing the LU decomposition of $Q$ requires $\frac 2 3 N^3$ multiplications.
\end{itemize}

\subsection{Running Time for 3d Ising Computations}
\label{app:runningtime}

For the 3d Ising computations in this work, we have 
\be
d&\approx& \Lambda+\frac 5 2 \numax\ \textrm{(for large enough $\ell$)},\nn\\
m&=&1\textrm{ or }2,\nn\\
J&\approx&\textrm{number of included spins}.
\ee
A full analysis of the complexity in $\Lambda$ would require determining the correct asymptotics of each quantity (including \SDPB\ parameters like \texttt{precision}), which may depend on the computation at hand.  In table~\ref{tab:runtimes}, we simply report average runtimes and approximate multiplications/iteration for the choices given in appendix~\ref{sec:parameters}.
\begin{table}[!ht]
\centering
\begin{tabular}{c|c|c|c|c} 
solver & $\Lambda$ & runtime (dual feasible) & runtime (dual infeasible) & mul./iter.\\
\hline
\SDPB & 19 & 3.5 & 0.9 & $2 \x 10^{8}$\\
\SDPB & 27 & 32 & 7.6 & $1\x 10^{9}$ \\
\SDPB & 35 & 190 & 40 & $5\x 10^9$ \\
\SDPB & 43 & 810 & 260 & $2 \x 10^{10}$ \\
\SDPAGMP & 19 & $\sim 300$ & $\sim 300$ & $1\x10^{11}$\\
\SDPAGMP & 27 & -- & -- & $1\x10^{12}$ \\
\SDPAGMP & 35 & -- & -- & $8\x10^{12}$ \\
\SDPAGMP & 43 & -- & -- & $5\x 10^{13}$
\end{tabular}
\caption{Runtimes for a single feasibility computation in the 3d Ising CFT using the correlators $\{\<\s\s\s\s\>, \<\s\s\e\e\>, \<\e\e\e\e\>\}$, as described in \cite{Kos:2014bka} and appendix~\ref{sec:parameters}.  Average runtimes are different depending on whether a spectrum is disallowed (dual feasible) or allowed (dual infeasible).  All times are in CPU-hours (for \SDPB, this means the actual runtime is multiplied by \texttt{maxThreads}, which was 16 for most of the computations in this work). Approximate \SDPAGMP\ times are from \cite{Kos:2014bka}.  All computations were performed on 3.3GHz 64-bit Intel Xeon Sandy Bridge processors.  The column ``mul./iter." gives the approximate number of multiplications per iteration, calculated according to the discussion in subsection~\ref{sec:complexity}.  (To estimate running time from the number of multiplications per iteration, one needs to take into account \texttt{precision}, which also increases with $\Lambda$.)  The \SDPAGMP\ computations with $\Lambda>19$ have not been attempted.}
\label{tab:runtimes}
\end{table}

\clearpage
\bibliography{Biblio}{}
\bibliographystyle{utphys}

\end{document}